\documentclass[11pt,pdftex,a4paper]{article}%
\pdfoutput=1
\usepackage{mathptmx}       
\usepackage{helvet}         
\usepackage{courier}        
\usepackage{type1cm}        
\usepackage{makeidx}         
\usepackage{hyperref}
\usepackage{graphicx}
\usepackage{verbatim}
\usepackage{amssymb}
\usepackage{lipsum}
\usepackage{multicol}
\usepackage{amsmath}
\usepackage{float}
\usepackage{amsthm}
\usepackage{tikz}
\usepackage{color}
\usepackage{subfig}
\usepackage{cite}
\newif\ifcode
\codefalse
\usepackage{macros}
\usepackage[ruled]{algorithm}
\usepackage{algpseudocode}
\usepackage{ioa_code}

\newtheorem{theorem}{Theorem}

\newtheorem{claim}[theorem]{Claim}

\newtheorem{definition}{Definition}
\newtheorem{lemma}[theorem]{Lemma}
\newtheorem{observation}[theorem]{Observation}

\newcounter{linenumber}

\newcommand{\true}{\mathit{true}}
\newcommand{\false}{\mathit{false}}

\newcommand{\remove}[1]{}

\newcommand{\Wset}{\textit{Wset}}
\newcommand{\Rset}{\textit{Rset}}
\newcommand{\Dset}{\textit{Dset}}

\newcommand{\txns}{\textit{txns}}

\newcommand{\Read}{\textit{read}}
\newcommand{\Write}{\textit{write}}
\newcommand{\TryC}{\textit{tryC}}

\newcommand{\ignore}[1]{}

\begin{document}
\bibliographystyle{abbrv}

\title{Progressive Transactional Memory in Time and Space}

\author{
Petr Kuznetsov$^1$~~~Srivatsan Ravi$^2$ \\
$^1$\normalsize T\'el\'ecom ParisTech \\
$^2$\normalsize TU Berlin
}

\ignore{
\author{
Petr Kuznetsov\inst{1} \and
Sathya Peri\inst{2}
}

\institute{
T\'el\'ecom ParisTech,  petr.kuznetsov@telecom-paristech.fr \and
IIT Patna, sathya@iitp.ac.in
}

\remove{
\authorinfo{Petr Kuznetsov}
           {Telekom Innovation Lab, TU Berlin, petr.kuznetsov@tu-berlin.de}
\authorinfo{Sathya Peri}
           {IIT Patna, sathya@iitp.ac.in}
}
}

\maketitle
\begin{abstract}
Transactional memory (TM) allows concurrent processes to
organize sequences of operations on shared \emph{data items} into atomic
transactions. 
A transaction may commit, in which case it appears to have executed sequentially or it may \emph{abort}, in which case no
data item is updated.

The TM programming paradigm emerged as an alternative to conventional
fine-grained locking techniques, offering 
ease of programming and compositionality.   
Though typically themselves implemented using locks, TMs
hide the inherent issues of lock-based synchronization behind a nice
transactional programming interface.

In this paper, we explore inherent time and space complexity of lock-based TMs,
with a focus of the most popular class of \emph{progressive} lock-based
TMs. 
We derive that a progressive TM 
might enforce a read-only transaction to perform a quadratic (in the number
of the data items it reads) number of steps and access a linear number
of distinct memory locations, closing the question of
inherent cost of \emph{read validation} in TMs.
We then show that the total number of \emph{remote memory references} (RMRs) 
that take place in an execution of a progressive TM in which $n$
concurrent processes perform transactions on a single data item might
reach $\Omega(n \log n)$, which appears to be the first RMR complexity
lower bound for transactional memory.
\end{abstract}

%
\section{Introduction}
\label{sec:intro}
Transactional memory (TM) allows concurrent processes to
organize sequences of operations on shared \emph{data items} into atomic
transactions. 
A transaction may \emph{commit}, in which case it appears to have executed sequentially or it may \emph{abort}, in which case no
data item is updated.
The user can therefore design software having only sequential
semantics in mind and let the TM take care of handling
\emph{conflicts} (concurrent reading and writing to the same data item) 
resulting from concurrent executions.
Another benefit of transactional memory over conventional lock-based concurrent
programming is \emph{compositionality}: it allows the programmer to
easily compose multiple operations on multiple objects into atomic
units, which is very hard to achieve using locks directly.      
Therefore, while still typically \emph{implemented} using locks, 
TMs hide the inherent issues of lock-based programming behind
an easy-to-use and compositional transactional interface. 
%

At a high level, a TM implementation must ensure that 
transactions are \emph{consistent} with some sequential execution.
A natural consistency criterion is 
\emph{strict serializability}~\cite{Pap79-serial}: 
all committed transactions appear to execute sequentially in some
total order respecting the timing of non-overlapping transactions.
The stronger criterion of \emph{opacity}~\cite{tm-book}, 
guarantees that \emph{every} transaction (including aborted and
incomplete ones) observes a view that is consistent with the
\emph{same} sequential execution, which implies that no transaction
would expose a pathological behavior, not predicted by the sequential
program, such as division-by-zero or
infinite loop. 

Notice that a TM implementation in which every transaction is aborted is trivially opaque, but not very useful.
Hence, the TM must satisfy some \emph{progress} guarantee specifying the conditions under
which a transaction is allowed to abort.
It is typically expected that a transaction aborts only because of
\emph{data conflicts} with a concurrent one, e.g., when they
are both trying to access the same data item and at least one of the
transactions is trying to update it. 
This progress guarantee, captured formally by the criterion of
\emph{progressiveness}~\cite{GK09-progressiveness}, is satisfied by
most TM implementations today~\cite{HLM+03,DSS06,norec}. 

There are two design principles which state-of-the-art TM~\cite{DSS06,norec,nztm,fraser,HLM+03, DStransaction06} 
implementations adhere to: \emph{read invisibility}~\cite{DS10-rwlock,AH13-inv} and
\emph{disjoint-access parallelism}~\cite{israeli-disjoint,AHM09}. Both are assumed to decrease the
chances of a transaction to encounter a data conflict and, thus,
improve performance of progressive TMs.   
Intuitively, reads performed by a TM are invisible if they do not
modify the shared memory used by the TM implementation and, thus, do
not affect other transactions.   
A disjoint-access parallel (DAP) TM ensures that transaction accessing
disjoint data sets do not contend on the shared memory and, thus, may
proceed independently. 
As was earlier observed~\cite{tm-book}, the combination of these principles incurs
some inherent costs, and the main motivation of this paper is to
explore these costs.    

Intuitively, the overhead invisible read may incur comes from the need
of \emph{validation}, \emph{i.e.}, ensuring that read data items
have not been updated when the transaction completes.
Our first result (Section~\ref{sec:invlb}) is that a read-only transaction in an opaque TM featured with \emph{weak} DAP
and \emph{weak} invisible reads must \emph{incrementally} validate
every next read operation. This results in a quadratic  (in the size of the transaction's read
set) step-complexity lower bound.
Informally, weak DAP means that two transactions
encounter a memory race only if their data 
sets are connected in the \emph{conflict graph}, capturing 
data-set overlaps among all concurrent transactions.
Weak read invisibility allows read operations of a transaction
$T$ to be ``visible'' only if $T$ is concurrent with
another transaction.
The lower bound is derived for \emph{minimal} progressiveness, 
where transactions are guaranteed to commit only if they run sequentially. 
Our result improves the lower bound~\cite{GK09-progressiveness,tm-book} derived for \emph{strict-data
partitioning} (a very strong version of DAP) and (strong) invisible reads.  

\ignore{
This result makes the assumption
that the shared memory is \emph{strict data partitioned}, \emph{i.e.}, the set of base objects used by the TM is split into 
disjoint sets, each storing information only about a single data item.
In this paper, we improve on this result. We first prove a quadratic (in the size of the transaction's read set) step-complexity lower bound opaque TMs that provide
sequential TM-progress while assuming weak invisible reads and a weaker (than strict data-partitioning) 
variant of DAP.
}
Our second result is that, under weak DAP and weak read invisibility,
a strictly serializable TM must have a read-only transaction that accesses a 
linear (in the size of the transaction's read set) number of distinct
memory locations in the course of performing its last read operation. 
Naturally, this space lower bound also applies to opaque TMs.

We then turn our focus to \emph{strongly progressive}
TMs~\cite{tm-book} that, in addition to progressiveness, 
ensures that \emph{not all} concurrent transactions conflicting over a single
data item abort.   
In Section~\ref{sec:rmr}, we prove that in any strongly progressive strictly serializable TM
implementation that accesses the shared memory with \emph{read}, \emph{write} and \emph{conditional}
primitives, such as \emph{compare-and-swap} and
\emph{load-linked/store-conditional}, 
the total number of \emph{remote memory references} (RMRs) 
that take place in an execution of a progressive TM in which $n$
concurrent processes perform transactions on a single data item might
reach $\Omega(n \log n)$. 
The result is obtained via a reduction to an analogous lower bound for
mutual exclusion~\cite{rmr-mutex}.
In the reduction, we show that any TM with the above properties can be
used to implement a \emph{deadlock-free} mutual exclusion, employing
transactional operations on only one data item and incurring a constant RMR overhead. 
The lower bound applies to RMRs in both the \emph{cache-coherent (CC)} and
\emph{distributed shared memory (DSM)} models, and it appears to be the first RMR complexity
lower bound for transactional memory.

\ignore{
has an execution
that incurs $\Omega(n\log n)$ ($n$ is the total number of processes) \emph{remote memory references (RMR)}~\cite{anderson-90-tpds}
even if it accesses just a single data item.
The lower bound applies to both the \emph{cache-coherent (CC)} and \emph{distributed shared memory (DSM)} models.
We prove this result by reduction to \emph{mutual exclusion}~\cite{rmr-mutex}. 
The mutual exclusion task involves sharing a critical resource among processes such that at most one process has access to
the resource at any prefix of the execution.
In \cite{rmr-mutex}, it established that any \emph{deadlock-free} (if every process takes infinitely many steps,
some process eventually accesses the resource) mutual exclusion implementation
from read, write and conditional primitives has an execution whose RMR cost is $\Omega(n\log n)$.
We show that, any strictly serializable
strongly progressive TM implementation $M$ implies a \emph{deadlock-free} mutual exclusion implementation $L(M)$ such that the RMR complexity of $M$
is within a constant factor of the RMR complexity of $L(M)$, thus proving the lower bound.

\vspace{2mm}\noindent\textbf{Roadmap.}
Sections~\ref{sec:model} and \ref{sec:prel} define our model and the classes of TMs 
considered in this paper.
Section~\ref{sec:invlb} contains our step-complexity lower bound and 
Section~\ref{sec:rmr} describes our lower bound on RMR complexity for strongly progressive TMs.
In Section~\ref{sec:related}, we discuss the related
work and in Section~\ref{sec:disc},  concluding remarks.
}

%
\section{Model}
\label{sec:model}
\vspace{1mm}\noindent\textbf{TM interface.}
A \emph{transactional memory} (in short, \emph{TM})
supports \emph{transactions} for
reading and writing on a finite set of data items,
referred to as \emph{t-objects}.
Every transaction $T_k$ has a unique identifier $k$. We assume no bound on the size of a t-object, \emph{i.e.}, 
the cardinality on the set $V$ of possible different values a t-object can have.
A transaction $T_k$ may contain the following \emph{t-operations},
each being a matching pair of an \emph{invocation} and a \emph{response}:
$\Read_k(X)$ returns a value in some domain $V$ (denoted $\Read_k(X) \rightarrow v$)
or a special value $A_k\notin V$ (\emph{abort});
$\Write_k(X,v)$, for a value $v \in V$,
returns \textit{ok} or $A_k$;
$\TryC_k$ returns $C_k\notin V$ (\emph{commit}) or $A_k$.

\vspace{1mm}\noindent\textbf{Implementations.}
We assume an asynchronous shared-memory system in which a set of $n>1$ processes $p_1,\ldots , p_n$ communicate by
applying \emph{operations} on shared \emph{objects}.
An object is an instance of an \emph{abstract data type} which specifies a set of operations that provide the only means to
manipulate the object.
An \emph{implementation} of an object type $\tau$ provides a specific data-representation of $\tau$ by 
applying \emph{primitives} on
shared \emph{base objects}, each of which is assigned an initial value and a set of algorithms $I_1(\tau),\ldots , I_n(\tau)$,
one for each process.
We assume that these primitives are \emph{deterministic}.
Specifically, a TM \emph{implementation} provides processes with algorithms
for implementing $\Read_k$, $\Write_k$ and $\TryC_k()$
of a transaction $T_k$ by \emph{applying} \emph{primitives} from a set of shared \emph{base objects}.
We assume that processes issue transactions sequentially, \emph{i.e.}, a process starts a new transaction only after the previous
transaction is committed or aborted.
A primitive is a generic \emph{read-modify-write} (\emph{RMW}) procedure applied to a base object~\cite{G05,Her91}.
It is characterized by a pair of functions $\langle g,h \rangle$:
given the current state of the base object, $g$ is an \emph{update function} that
computes its state after the primitive is applied, while $h$ 
is a \emph{response function} that specifies the outcome of the primitive returned to the process.
A RMW primitive is \emph{trivial} if it never changes the value of the base object to which it is applied.
Otherwise, it is \emph{nontrivial}.
An RMW primitive $\langle g,h \rangle$ is \emph{conditional} if there exists $v$, $w$ such that
$g(v,w)=v$ and there exists $v$, $w$ such that
$g(v,w)\neq v$~\cite{cond-04}.
For \emph{e.g}, \emph{compare-and-swap (CAS)}
and \emph{load-linked/store-conditional (LL/SC} are nontrivial conditional RMW primitives
while \emph{fetch-and-add} is an example of a nontrivial RMW primitive that is not conditional.

\vspace{1mm}\noindent\textbf{Executions and configurations.}
An \emph{event} of a process $p_i$ (sometimes we say \emph{step} of $p_i$)
is an invocation or response of an operation performed by $p_i$ or a 
rmw primitive $\langle g,h \rangle$ applied by $p_i$ to a base object $b$
along with its response $r$ (we call it a \emph{rmw event} and write $(b, \langle g,h\rangle, r,i)$).
A \emph{configuration} specifies the value of each base object and 
the state of each process.
The \emph{initial configuration} is the configuration in which all 
base objects have their initial values and all processes are in their initial states.

An \emph{execution fragment} is a (finite or infinite) sequence of events.
An \emph{execution} of an implementation $I$ is an execution
fragment where, starting from the initial configuration, each event is
issued according to $I$ and each response of a rmw event $(b, \langle
g,h\rangle, r,i)$ matches the state of $b$ resulting from all
preceding events.
An execution $E\cdot E'$, denoting the concatenation of $E$ and $E'$,
is an \emph{extension} of $E$ and we say that $E'$ \emph{extends} $E$.

Let $E$ be an execution fragment.
For every transaction identifier $k$,
$E|k$ denotes the subsequence of $E$ restricted to events of
transaction $T_k$.
If $E|k$ is non-empty,
we say that $T_k$ \emph{participates} in $E$, else we say $E$ is \emph{$T_k$-free}.
Two executions $E$ and $E'$ are \emph{indistinguishable} to a set $\mathcal{T}$ of transactions, if
for each transaction $T_k \in \mathcal{T}$, $E|k=E'|k$.
A TM \emph{history} is the subsequence of an execution consisting of the invocation and 
response events of t-operations.

The \emph{read set} (resp., the \emph{write set}) of a transaction $T_k$ in an execution $E$,
denoted $\Rset(T_k)$ (and resp. $\Wset(T_k)$), is the set of t-objects on which $T_k$ invokes reads (and resp. writes) in $E$.
The \emph{data set} of $T_k$ is $\Dset(T_k)=\Rset(T_k)\cup\Wset(T_k)$.
A transaction is called \emph{read-only} if $\Wset(T_k)=\emptyset$; \emph{write-only} if $\Rset(T_k)=\emptyset$ and
\emph{updating} if $\Wset(T_k)\neq\emptyset$.
Note that, in our TM model, the data set of a transaction is not known apriori, \emph{i.e.}, at the start of the transaction
and it is identifiable only by the set of data items the transaction has invoked a read or write on in the given execution.

\vspace{1mm}\noindent\textbf{Transaction orders.}
Let $\txns(E)$ denote the set of transactions that participate in $E$.
An execution $E$ is \emph{sequential} if every invocation of
a t-operation is either the last event in the history $H$ exported by $E$ or
is immediately followed by a matching response.
We assume that executions are \emph{well-formed}:
no process invokes a new operation before
the previous operation returns.
Specifically, we assume that for all $T_k$, $E|k$ begins with the invocation of a t-operation, is
sequential and has no events after $A_k$ or $C_k$.
A transaction $T_k\in \txns(E)$ is \emph{complete in $E$} if
$E|k$ ends with a response event.
The execution $E$ is \emph{complete} if all transactions in $\txns(E)$
are complete in $E$.
A transaction $T_k\in \txns(E)$ is \emph{t-complete} if $E|k$
ends with $A_k$ or $C_k$; otherwise, $T_k$ is \emph{t-incomplete}.
$T_k$ is \emph{committed} (resp., \emph{aborted}) in $E$
if the last event of $T_k$ is $C_k$ (resp., $A_k$).
The execution $E$ is \emph{t-complete} if all transactions in
$\txns(E)$ are t-complete.

For transactions $\{T_k,T_m\} \in \txns(E)$, we say that $T_k$ \emph{precedes}
$T_m$ in the \emph{real-time order} of $E$, denoted $T_k\prec_E^{RT} T_m$,
if $T_k$ is t-complete in $E$ and
the last event of $T_k$ precedes the first event of $T_m$ in $E$.
If neither $T_k\prec_E^{RT} T_m$ nor $T_m\prec_E^{RT} T_k$,
then $T_k$ and $T_m$ are \emph{concurrent} in $E$.
An execution $E$ is \emph{t-sequential} if there are no concurrent
transactions in $E$.

\vspace{1mm}\noindent\textbf{Contention.}
We say that a configuration $C$ after an execution $E$ is \emph{quiescent} (and resp. \emph{t-quiescent})
if every transaction $T_k \in \ms{txns}(E)$ is complete (and resp. t-complete) in $C$.
If a transaction $T$ is incomplete in an execution $E$, it has exactly one \emph{enabled} event, 
which is the next event the transaction will perform according to the TM implementation.
Events $e$ and $e'$ of an execution $E$  \emph{contend} on a base
object $b$ if they are both events on $b$ in $E$ and at least 
one of them is nontrivial (the event is trivial (and resp. nontrivial) if it is the application of a trivial 
(and resp. nontrivial) primitive).
We say that a transaction $T$ is \emph{poised to apply an event $e$ after $E$} 
if $e$ is the next enabled event for $T$ in $E$.
We say that transactions $T$ and $T'$ \emph{concurrently contend on $b$ in $E$} 
if they are each poised to apply contending events on $b$ after $E$.

We say that an execution fragment $E$ is \emph{step contention-free for t-operation $op_k$} if the events of $E|op_k$ 
are contiguous in $E$.
We say that an execution fragment $E$ is \emph{step contention-free for $T_k$} if the events of $E|k$ are contiguous in $E$.
We say that $E$ is \emph{step contention-free} if $E$ is step contention-free for all transactions that participate in $E$.
\section{TM classes}
\label{sec:prel}
\vspace{1mm}\noindent\textbf{TM-correctness.}
We say that $\Read_k(X)$ is \emph{legal} in a t-sequential execution $E$ if it returns the
\emph{latest written value} of $X$, and $E$ is \emph{legal}
if every $\Read_k(X)$ in $H$ that does not return $A_k$ is legal in $E$.

%
A finite history $H$ is \emph{opaque} if there
is a legal t-complete t-sequential history $S$,
such that
(1) for any two transactions $T_k,T_m \in \txns(H)$,
if $T_k \prec_H^{RT} T_m$, then $T_k$ precedes $T_m$ in $S$, and
(2) $S$ is equivalent to a \emph{completion} of $H$.

A finite history $H$ is \emph{strictly serializable} if there
is a legal t-complete t-sequential history $S$,
such that
(1) for any two transactions $T_k,T_m \in \txns(H)$,
if $T_k \prec_H^{RT} T_m$, then $T_k$ precedes $T_m$ in $S$, and
(2) $S$ is equivalent to $\ms{cseq}(\bar H)$, where $\bar H$ is some
completion of $H$ and $\ms{cseq}(\bar H)$ is the subsequence of $\bar H$ reduced to
committed transactions in $\bar H$.

We refer to $S$ as an opaque (and resp. strictly serializable) \emph{serialization} of $H$.

\vspace{1mm}\noindent\textbf{TM-liveness.}
We say that a TM implementation $M$ provides \emph{interval-contention free (ICF)} TM-liveness
if for every finite execution $E$ of $M$ such that the configuration after $E$ is quiescent, 
and every transaction $T_k$ that applies the invocation of a t-operation $op_k$ immediately after $E$, 
the finite step contention-free extension for $op_k$ contains a matching response.

\vspace{1mm}\noindent\textbf{TM-progress.}
We say that a TM implementation provides \emph{sequential TM-progress} (also called \emph{minimal progressiveness}~\cite{tm-book})
if every transaction
running step contention-free from a t-quiescent configuration commits within a finite number of steps.

We say that transactions $T_i,T_j$ \emph{conflict} in an execution $E$ on a t-object $X$ if
$X\in\Dset(T_i)\cap\Dset(T_j)$,  and $X\in\Wset(T_i)\cup\Wset(T_j)$.

A TM implementation $M$ provides \emph{progressive} TM-progress (or \emph{progressiveness}) 
if for every execution $E$ of $M$ and every transaction $T_i \in \ms{txns}(E)$ that returns $A_i$ in $E$, 
there exists a transaction $T_k \in \ms{txns}(E)$ such that $T_k$ and $T_i$ are concurrent and conflict in $E$~\cite{tm-book}. 

Let $CObj_H(T_i)$ denote the set of t-objects over which transaction $T_i \in \ms{txns}(H)$ conflicts with any other 
transaction in history $H$,
\emph{i.e.}, $X \in CObj_H(T_i)$, \emph{iff} there exist transactions $T_i$ and $T_k$
that conflict on $X$ in $H$.
Let $ Q\subseteq \ms{txns}(H)$ and $CObj_H(Q)=\bigcup\limits_{ T_i \in Q} CObj_H(T_i)$.

Let $\ms{CTrans}(H)$ denote the set of non-empty subsets of $\ms{txns}(H)$ such that a set $Q$ is in $\ms{CTrans}(H)$ 
if no transaction in $Q$ conflicts with a transaction not in $Q$.
\begin{definition}
\label{def:sprog}
A TM implementation $M$ is \emph{strongly progressive} if $M$ is weakly progressive 
and for every history $H$ of $M$ and for every set $Q \in \ms{CTrans}(H)$ such that $|CObj_{H}(Q)| \leq 1$, 
some transaction in $Q$ is not aborted in $H$.
\end{definition}
\vspace{1mm}\noindent\textbf{Invisible reads.}
A TM implementation $M$ uses \emph{invisible reads} if for every execution $E$ of $M$ and for every read-only transaction
$T_k\in \ms{txns}(E)$, $E|k$ does not contain any nontrivial events. 

In this paper, we introduce a definition of \emph{weak} invisible reads.
For any execution $E$ and any t-operation $\pi_k$ invoked by some transaction $T_k\in \ms{txns}(E)$,
let $E|\pi_k$ denote the subsequence of $E$ restricted to events of $\pi_k$ in $E$.

We say that a TM implementation $M$ satisfies \emph{weak invisible reads}
if for any execution $E$ of $M$ and every transaction $T_k\in \ms{txns}(E)$; $\Rset(T_k)\neq \emptyset$ that is
not concurrent with any transaction $T_m\in \ms{txns}(E)$, $E|\pi_k$ does not contain any nontrivial events, where $\pi_k$ is
any t-read operation invoked by $T_k$ in $E$.
%

\vspace{1mm}\noindent\textbf{Disjoint-access parallelism (DAP).}
%
Let $\tau_{E}(T_i,T_j)$ be the set of transactions ($T_i$ and $T_j$ included)
that are concurrent to at least one of $T_i$ and $T_j$ in $E$.
Let $G(T_i,T_j,E)$ be an undirected graph whose vertex set is $\bigcup\limits_{T \in \tau_{E}(T_i,T_j)} \Dset(T)$
and there is an edge
between t-objects $X$ and $Y$ \emph{iff} there exists $T \in \tau_{E}(T_i,T_j)$ such that 
$\{X,Y\} \in \Dset(T)$.
We say that $T_i$ and $T_j$ are \emph{disjoint-access} in $E$
if there is no path between a t-object in $\Dset(T_i)$ and a t-object in $\Dset(T_j)$ in $G(T_i,T_j,E)$.
A TM implementation $M$ is \emph{weak disjoint-access parallel (weak DAP)} if, for
all executions $E$ of $M$, 
transactions $T_i$ and $T_j$ 
concurrently contend on the same base object in $E$ only if   
$T_i$ and $T_j$ are not disjoint-access in $E$ or there exists a t-object $X \in \Dset(T_i) \cap \Dset(T_j)$~\cite{AHM09,PFK10}.
\begin{lemma}
\label{lm:dap}
(\cite{AHM09},\cite{TM-WF14})
Let $M$ be any weak DAP TM implementation.
Let $\alpha\cdot \rho_1 \cdot \rho_2$ be any execution of $M$ where
$\rho_1$ (and resp. $\rho_2$) is the step contention-free
execution fragment of transaction $T_1 \not\in \ms{txns}(\alpha)$ (and resp. $T_2 \not\in \ms{txns}(\alpha)$) 
and transactions $T_1$, $T_2$ are disjoint-access in $\alpha\cdot \rho_1 \cdot \rho_2$. 
Then, $T_1$ and $T_2$ do not contend on any base object in $\alpha\cdot \rho_1 \cdot \rho_2$.
\end{lemma}
\begin{figure*}[t]
\begin{center}
	\subfloat[$R_{\phi}(X_{i})$ must return $nv$ by strict serializability \label{sfig:inv-1}]{\scalebox{0.6}[0.6]{\begin{tikzpicture}
\node (r1) at (3,0) [] {};
\node (r2) at (7.7,0) [] {};

\node (w1) at (-2,0) [] {};

\draw (r1) node [below] {\normalsize {$R_{\phi}(X_1) \cdots R_{\phi}(X_{i-1})$}};
\draw (r1) node [above] {\normalsize {$i-1$ t-reads}};

\draw (r2) node [above] {\normalsize {$R_{\phi}(X_i)\rightarrow nv$}};

\draw (w1) node [above] {\normalsize {$W_i(X_i,nv)$}}; 
\draw (w1) node [below] {\normalsize {$T_i$ commits}};

\begin{scope}   
\draw [|-|,thick] (0,0) node[left] {$T_{\phi}$} to (6,0);
\draw [|-|,thick] (6.5,0) node[left] {} to (9,0);
\draw [-,dotted] (0,0) node[left] {} to (9,0);
\end{scope}
\begin{scope}   
\draw [|-|,thick] (-3,0) node[left] {$T_i$} to (-1,0);
\end{scope}
\end{tikzpicture}}}
        \\
        \vspace{2mm}
	\subfloat[$T_i$ does not observe any conflict with $T_{\phi}$ \label{sfig:inv-2}]{\scalebox{0.6}[0.6]{\begin{tikzpicture}
\node (r1) at (3,0) [] {};
\node (r3) at (12.2,0) [] {};


\node (w2) at (7.5,-2) [] {};

\draw (r1) node [below] {\small {$R_{\phi}(X_1) \cdots R_{\phi}(X_{i-1})$}};
\draw (r1) node [above] {\small {$i-1$ t-reads}};

\draw (w2) node [above] {\small {$W_{i}(X_{i},nv)$}}; 
\draw (w2) node [below] {\small {$T_{i}$ commits}};

\draw (r3) node [above] {\small {$R_{\phi}(X_{i})\rightarrow nv$}};
\draw (r3) node [below] {\small {new value}};

\begin{scope}   
\draw [|-|,thick] (0,0) node[left] {$T_{\phi}$} to (6,0);
\draw [|-|,dotted] (0,0) node[left] {$T_{\phi}$} to (13.5,0);
\draw [|-|,thick] (11,0) node[left] {} to (13.5,0);
\end{scope}
\begin{scope}   
\draw [|-|,thick] (6.5,-2) node[left] {$T_i$} to (9,-2);
\end{scope}
\end{tikzpicture}}}
	\caption{Executions in the proof of Lemma~\ref{lm:readdap}; By weak DAP, $T_{\phi}$ cannot distinguish this from the execution in Figure~\ref{sfig:inv-1}
        \label{fig:indis}} 
\end{center}
\end{figure*}
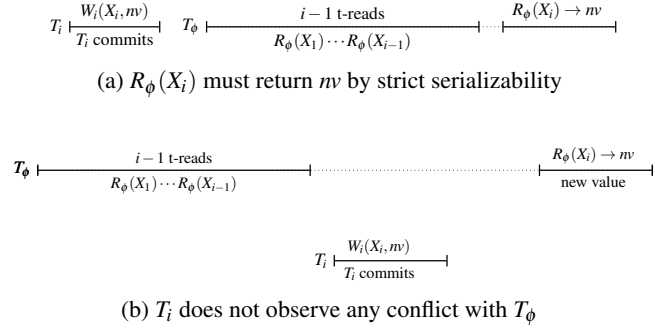
%
%
%
\section{Time and space complexity of sequential TMs}
\label{sec:invlb}
In this section, we prove that
(1) that a read-only transaction in an opaque TM featured with \emph{weak} DAP
and \emph{weak} invisible reads must \emph{incrementally} validate
every next read operation, and (2) a strictly serializable TM (under weak DAP and weak read invisibility), must 
have a read-only transaction that accesses a 
linear (in the size of the transaction's read set) number of distinct
base objects in the course of performing its last t-read and tryCommit operations.

We first prove the following lemma concerning strictly serializable weak DAP TM implementations.
\begin{lemma}
\label{lm:readdap}
Let $M$ be any strictly serializable, weak DAP TM implementation that provides sequential TM-progress.
Then, for all $i\in \mathbb{N}$, $M$ has an execution of the form $\pi^{i-1}\cdot \rho^i\cdot \alpha^i$ where,
\begin{itemize}
\item
$\pi^{i-1}$ is the complete step contention-free execution of read-only transaction $T_{\phi}$ that performs
$(i-1)$ t-reads: $\Read_{\phi}(X_1)\cdots \Read_{\phi}(X_{i-1})$,
\item
$\rho^i$ is the t-complete step contention-free execution of a transaction $T_{i}$
that writes $nv_i\neq v_i$ to $X_i$ and commits,
\item
$\alpha_i$ is the complete step contention-free execution fragment of $T_{\phi}$ that performs its $i^{th}$ t-read:
$\Read_{\phi}(X_i) \rightarrow nv_i$.
\end{itemize}
\end{lemma}
\begin{proof}
By sequential TM-progress, $M$ has an execution of the form $\rho^i\cdot \pi^{i-1}$.
Since $\Dset(T_k) \cap \Dset(T_{i}) =\emptyset$ in $\rho^i\cdot \pi^{i-1}$,
by Lemma~\ref{lm:dap}, transactions $T_{\phi}$ and $T_i$ do not contend
on any base object in execution $\rho^i\cdot \pi^{i-1}$.
Thus, $\rho^i\cdot \pi^{i-1}$ is also an execution of $M$.

By assumption of strict serializability, $\rho^i\cdot \pi^{i-1} \cdot \alpha_i$ is an execution
of $M$ in which the t-read of $X_i$ performed by $T_{\phi}$ must return $nv_i$.
But $\rho^i \cdot \pi^{i-1} \cdot \alpha_i$ is indistinguishable to $T_{\phi}$ from
$\pi^{i-1}\cdot \rho^i \cdot \alpha_i$.
Thus, $M$ has an execution of the form $\pi^{i-1}\cdot \rho^i \cdot \alpha_i$.
\end{proof}
\begin{theorem}
\label{th:iclb}
For every weak DAP TM implementation $M$ that provides ICF TM-liveness, sequential TM-progress and uses weak invisible reads,
\begin{enumerate}
\item[(1)]
If $M$ is opaque,
for every $m\in \mathbb{N}$,
there exists an execution $E$ of $M$ 
such that some transaction $T\in \ms{txns}(E)$ performs $\Omega(m^2)$ steps, where $m=|\Rset(T_k)|$.
\item[(2)]
if $M$ is strictly serializable,
for every $m\in \mathbb{N}$,
there exists an execution $E$ of $M$
such that some transaction $T_k\in \ms{txns}(E)$
accesses at least $m-1$ distinct base objects during the executions of the $m^{th}$ t-read operation
and $\TryC_k()$, where $m=|\Rset(T_k)|$.
\end{enumerate}
\end{theorem}
\begin{proof}
For all $i\in \{1,\ldots , m\}$, let $v$ be the initial value of t-object $X_i$.

($1$) Suppose that $M$ is opaque.
Let $\pi^{m}$ denote the complete step contention-free execution of a transaction
$T_{\phi}$ that performs ${m}$ t-reads: $\Read_{\phi}(X_1)\cdots \Read_{\phi}(X_{m})$
such that for all $i\in \{1,\ldots , m \}$, $\Read_{\phi}(X_i) \rightarrow v$.

By Lemma~\ref{lm:readdap}, for all $i\in \{2,\ldots, m\}$, $M$ has an execution of the form 
$E^{i}=\pi^{i-1}\cdot \rho^i \cdot \alpha_i$.

For each $i\in \{2,\ldots, m\}$, $j\in \{1,2\}$ and $\ell \leq (i-1)$, 
we now define an execution of the form  $\mathbb{E}_{j\ell}^{i}=\pi^{i-1}\cdot \beta^{\ell}\cdot \rho^i \cdot \alpha_j^i$
as follows:
\begin{itemize}
\item
$\beta^{\ell}$ is the t-complete step contention-free execution fragment of a transaction $T_{\ell}$
that writes $nv_{\ell}\neq v$ to $X_{\ell}$ and commits
\item
$\alpha_1^i$ (and resp. $\alpha_2^i$) is the complete step contention-free execution fragment of 
$\Read_{\phi}(X_i) \rightarrow v$ (and resp. $\Read_{\phi}(X_i) \rightarrow A_{\phi}$).
\end{itemize}
\begin{claim}
\label{cl:ic2}
For all $i\in \{2,\ldots, m\}$ and $\ell \leq (i-1)$, $M$ has an execution of the form $\mathbb{E}_{1\ell}^{i}$ or 
$\mathbb{E}_{2\ell}^{i}$.
\end{claim}
\begin{proof}
For all $i \in \{2,\ldots, m\}$, $\pi^{i-1}$
is an execution of $M$.
By assumption of weak invisible reads and sequential TM-progress, $T_{{\ell}}$ must be committed in $\pi^{i-1}\cdot \rho^{\ell}$
and $M$ has an execution of the form $\pi^{i-1}\cdot \beta^{\ell}$.
By the same reasoning, since $T_i$ and $T_{\ell}$ have disjoint data sets,
$M$ has an execution of the form $\pi^{i-1}\cdot\beta^{\ell}\cdot \rho^i$.

Since the configuration after $\pi^{i-1}\cdot\beta^{\ell}\cdot \rho^i$ is quiescent,
by ICF TM-liveness, $\pi^{i-1}\cdot\beta^{\ell}\cdot \rho^i$ extended with $\Read_{\phi}(X_i)$
must return a matching response.
If $\Read_{\phi}(X_i) \rightarrow v_i$, then clearly $\mathbb{E}_{1}^{i}$
is an execution of $M$ with $T_{\phi}, T_{i-1}, T_i$ being a valid serialization
of transactions.
If $\Read_{\phi}(X_i) \rightarrow A_{\phi}$, the same serialization
justifies an opaque execution.

Suppose by contradiction that there exists an execution of $M$ such that
$\pi^{i-1}\cdot\beta^{\ell}\cdot \rho^i$ is extended with the complete execution
of $\Read_{\phi}(X_i) \rightarrow r$; $r \not\in \{A_{\phi},v\}$. 
The only plausible case to analyse is when $r=nv$.
Since $\Read_{\phi}(X_i)$ returns the value of $X_i$ updated by $T_i$, 
the only possible serialization for transactions is $T_{\ell}$, $T_i$, $T_{\phi}$; but $\Read_{\phi}(X_{\ell})$
performed by $T_k$ that returns the initial value $v$
is not legal in this serialization---contradiction.
\end{proof}
We now prove that, for all $i\in \{2,\ldots, m\}$, $j\in \{1,2\}$ and $\ell \leq (i-1)$, transaction $T_{\phi}$ must access
$(i-1)$ different base objects during the execution of $\Read_{\phi}(X_i)$ in the execution
$\pi^{i-1}\cdot \beta^{\ell}\cdot \rho^i \cdot \alpha_j^i$.

By the assumption of weak invisible reads,
the execution $\pi^{i-1}\cdot \beta^{\ell}\cdot \rho^i \cdot \alpha_j^i$
is indistinguishable to
transactions $T_{\ell}$ and $T_{i}$
from the execution ${\tilde \pi}^{i-1}\cdot \beta^{\ell}\cdot \rho^i \cdot \alpha_j^i$, where $\Rset(T_{\phi})=\emptyset$
in ${\tilde \pi}^{i-1}$.
But transactions $T_{\ell}$ and $T_{i}$ are disjoint-access in ${\tilde \pi}^{i-1}\cdot \beta^{\ell}\cdot \rho^i$ and by Lemma~\ref{lm:dap},
they cannot contend on the same base object in this execution.

Consider the $(i-1)$ different executions: 
$\pi^{i-1}\cdot\beta^{1}\cdot \rho^i$, $\ldots$, $\pi^{i-1}\cdot\beta^{i-1}\cdot \rho^i$.
For all $\ell, \ell' \leq (i-1)$;$\ell' \neq \ell$, 
$M$ has an execution of the form $\pi^{i-1}\cdot \beta^{\ell}\cdot \rho^i \cdot \beta^{\ell'}$
in which transactions $T_{\ell}$ and $T_{\ell'}$ access mutually disjoint data sets.
By weak invisible reads and Lemma~\ref{lm:dap}, the pairs of transactions $T_{\ell'}$, $T_{i}$ and $T_{\ell'}$, $T_{\ell}$
do not contend on any base object in this execution.
This implies that $\pi^{i-1}\cdot \beta^{\ell} \cdot \beta^{\ell'} \cdot \rho^i$ is an execution of $M$ in which
transactions $T_{\ell}$ and $T_{\ell'}$ each apply nontrivial primitives
to mutually disjoint sets of base objects in the execution fragments $\beta^{\ell}$ and $\beta^{\ell'}$ respectively 
(by Lemma~\ref{lm:dap}).

This implies that for any $j\in \{1,2\}$, $\ell \leq (i-1)$, the configuration $C^i$ after $E^i$ differs from the configurations
after $\mathbb{E}_{j\ell}^{i}$ only in the states of the base objects that are accessed in the fragment $\beta^{\ell}$.
Consequently, transaction $T_{\phi}$ must access at least $i-1$ different base objects
in the execution fragment $\pi_j^i$
to distinguish configuration $C^i$ from the configurations
that result after the $(i-1)$ different executions 
$\pi^{i-1}\cdot\beta^{1}\cdot \rho^i$, $\ldots$, $\pi^{i-1}\cdot\beta^{i-1}\cdot \rho^i$ respectively.

Thus, for all $i \in \{2,\ldots, m\}$, transaction $T_{\phi}$ must perform at least $i-1$ steps 
while executing the $i^{th}$ t-read in $\pi_{j}^i$ and $T_{\phi}$ itself must perform 
$\sum\limits_{i=1}^{m-1} i=\frac{m(m-1)}{2}$ steps.

($2$) Suppose that $M$ is strictly serializable, but not opaque.
Since $M$ is strictly serializable, by Lemma~\ref{lm:readdap}, it has an execution of the form 
$E=\pi^{m-1}\cdot \rho^{m} \cdot \alpha_m$.

For each $\ell \leq (i-1)$, we prove that $M$ has an execution of the form 
$E_{\ell}= \pi^{m-1}\cdot \beta^{\ell}\cdot \rho^m \cdot {\bar \alpha}^m$
where ${\bar \alpha}^m$ is the complete step contention-free execution fragment of $\Read_{\phi}(X_m)$ followed
by the complete execution of $\TryC_{\phi}$.
Indeed, by weak invisible reads, $\pi^{m-1}$ does not contain any nontrivial events
and the execution $\pi^{m-1}\cdot \beta^{\ell}\cdot \rho^m$ is indistinguishable to transactions
$T_{\ell}$ and $T_m$ from the executions ${\tilde \pi}^{m-1}\cdot \beta^{\ell}$ and ${\tilde \pi}^{m-1}\cdot\beta^{\ell} \cdot \rho^m$ 
respectively, where $\Rset(T_{\phi})=\emptyset$ in ${\tilde \pi}^{m-1}$.
Thus, applying Lemma~\ref{lm:dap}, transactions $\beta^{\ell} \cdot \rho^m$ do not contend
on any base object in the execution $\pi^{m-1}\cdot \beta^{\ell}\cdot \rho^m$. 
By ICF TM-liveness, $\Read_{\phi}(X_m)$ and $\TryC_{\phi}$ must return matching responses in the execution
fragment ${\bar \alpha}^m$ that extends $\pi^{m-1}\cdot \beta^{\ell}\cdot \rho^m$.
Consequently, for each $\ell \leq (i-1)$, $M$ has an execution of the form 
$E_{\ell}= \pi^{m-1}\cdot \beta^{\ell}\cdot \rho^m \cdot {\bar \alpha}^m$
such that transactions $T_{\ell}$ and $T_m$ do not contend on any base object.

Strict serializability of $M$ means that if $\Read_{\phi}(X_m)\rightarrow nv$ in the execution fragment ${\bar \alpha}^m$, 
then $\TryC_{\phi}$ must return $A_{\phi}$.
Otherwise if $\Read_{\phi}(X_m)\rightarrow v$ (i.e. the initial value of $X_m$), then
$\TryC_{\phi}$ may return $A_{\phi}$ or $C_{\phi}$.

Thus, as with ($1$), in the worst case, $T_{\phi}$ must access at least $m-1$ distinct base objects 
during the executions of
$\Read_{\phi}(X_m)$ and $\TryC_{\phi}$ to distinguish the configuration $C^i$ from the configurations
after the $m-1$ different executions
$\pi^{m-1}\cdot\beta^{1}\cdot \rho^m$, $\ldots$, $\pi^{m-1}\cdot\beta^{m-1}\cdot \rho^m$ respectively.
\end{proof}
\section{RMR complexity of strongly progressive TMs}
\label{sec:rmr}
\begin{algorithm}[t]
\caption{Mutual-exclusion object $L$ from a strongly progressive, strict serializable TM $M$; code for process $p_i$; $1\leq i \leq n$}
\label{alg:mutex-dsm}
\begin{algorithmic}[1]
  	\begin{multicols}{2}
  	{
  	\footnotesize
	\Part{Local variables}{
		\State bit $\ms{face}_i$, for each process $p_i$
	}\EndPart
	\Statex
	\Part{Shared objects}{
		\State strongly progressive, strictly 
		\State ~~serializable TM $M$
		\State ~~t-object $X$, initially $\bot$  
		\State ~~storing value $v \in \{[p_i, \ms{face}_i] \} \cup \{\bot\}$
				
		\State for each tuple $[p_i,\ms{face}_i]$
		\State ~~$\ms{Done}[p_i,\ms{face}_i] \in \{\true,\false\}$
		\State ~~$\ms{Succ}[p_i,\ms{face}_i] \in \{p_1,\ldots , p_n \} \cup \{\bot\}$
		\State for each $p_i$ and $j\in \{1,\ldots , n\}\setminus \{i\}$
		\State ~~$\ms{Lock}[p_i][p_j] \in \{\ms{locked},\ms{unlocked}\}$
	}\EndPart	
		
	\Statex	
	\Statex
	\Part{Function: \lit{func}()}{
		\State \textbf{atomic using $M$}
		
		 \State ~~~~$\ms{value}:= \lit{tx-read}(X)$
		 \State ~~~~$\lit{tx-write}(X,[p_i,\ms{face}_i])$
		\State \textbf{on abort Return} $\false$
		 \Return $\ms{value}$ \EndReturn
		 
	}\EndPart	
 	
 	\newpage
	\Part{Entry}{
		 \State $\ms{face}_i :=1- \ms{face}_i$
		 \State $\ms{Done}[p_i,\ms{face}_i].\lit{write}(\false)$ \label{line:entrydone1}
		 \State $\ms{Succ}[p_i,\ms{face}_i].\lit{write}(\bot)$ \label{line:entrysucc1}
		 \While{$(\ms{prev} \gets \lit{func})=\false$} \label{line:tm}
		      \State \textbf{no op}
		 \EndWhile \label{line:while}
		 \If{$\ms{prev} \neq \bot$}
		    \State $\ms{Lock}[p_i][\ms{prev}.pid].\lit{write}(\ms{locked})$ \label{line:entrylock}
		    \State $\ms{Succ}[\ms{prev}].\lit{write}(p_i)$ \label{line:entrysucc2}
		    \If{$\ms{Done}[\ms{prev}]=\false$} \label{line:entrydone2}
		      \While{$\ms{Lock}[p_i][\ms{prev}.pid]=\ms{unlocked}$} \label{line:entryunlock}
			\State \textbf{no op}
		      \EndWhile
		    \EndIf
		 \EndIf
		 \Return $ok$ \EndReturn
		 \State \Comment{Critical section}
		 
   	 }\EndPart
	\Statex
	
	\Part{Exit}{
		  \State $\ms{Done}[p_i,\ms{face}_i].\lit{write}(\true)$ \label{line:exitdone}
		  \State $\ms{Lock}[\ms{Succ}[p_i, \ms{face}_i]][p_i].\lit{write}(\ms{unlocked})$ \label{line:exitunlock}
		  \Return $ok$ \EndReturn
					
   	}\EndPart
		
	}
	\end{multicols}
  \end{algorithmic}
\end{algorithm}
In this section, we prove every strongly progressive strictly serializable TM that uses only 
read, write and \emph{conditional}
RMW primitives has an execution in which
in which $n$
concurrent processes perform transactions on a single data item and incur
$\Omega (\log n)$ \emph{remote memory references}~\cite{anderson-90-tpds}.

\vspace{1mm}\noindent\textbf{Remote memory references(RMR)~\cite{rmr-mutex}.}
In the \emph{cache-coherent (CC) shared memory}, each process maintains \emph{local}
copies of shared objects inside its cache, whose consistency is ensured by a coherence protocol.
Informally, we say that an access to a base object $b$ is \emph{remote} to a process $p$ and 
causes a \emph{remote memory reference (RMR)} if $p$'s cache contains a 
cached copy of the object that is out of date or \emph{invalidated}; otherwise the access is \emph{local}.

In the \emph{write-through (CC) protocol}, to read a base object $b$, process $p$ must have a cached copy of $b$ that
has not been invalidated since its previous read. Otherwise, $p$ incurs a RMR. 
To write to $b$, $p$ causes a RMR that invalidates all cached copies
of $b$ and writes to the main memory.

In the \emph{write-back (CC) protocol}, $p$ reads a base object $b$ without causing a RMR if it holds a cached copy of $b$
in shared or exclusive mode; otherwise the access of $b$ causes a RMR that (1) 
invalidates all copies of $b$ held in exclusive mode, and writing $b$ back to the main memory,
(2) creates a cached copy of $b$ in shared mode.
Process $p$ can write to $b$ without causing a RMR if it holds a copy of $b$ in exclusive mode; otherwise
$p$ causes a RMR that invalidates all cached copies of $b$ and creates a cached copy of $b$ in exclusive mode.

In the \emph{distributed shared memory (DSM)}, each register is forever assigned to a single process and it
\emph{remote} to the others. Any access of a remote register causes a RMR.

\vspace{1mm}\noindent\textbf{Mutual exclusion.}
The \emph{mutex object} supports two operations: \emph{Enter} and \emph{Exit}, both of which return the response $ok$.
We say that a process $p_i$ \emph{is in the critical section after an execution $\pi$} if
$\pi$ contains the invocation of $\lit{Enter}$ by
$p_i$ that returns $ok$, but does not contain a subsequent
invocation of $\lit{Exit}$ by $p_i$ in $\pi$. 
     
A mutual exclusion implementation satisfies the following properties:

(\emph{Mutual-exclusion}) After any execution
  $\pi$, there exists at most one process that is in the critical section. 

(\emph{Deadlock-freedom})  Let $\pi$ be any execution that contains the invocation of
  $\lit{Enter}$ by process $p_i$. Then, in every extension of $\pi$ in which every process takes infinitely many steps, some
  process is in the critical section.

(\emph{Finite-exit)} Every process completes the $\lit{Exit}$ operation within a finite number of steps.
%
\subsection{Mutual exclusion from a strongly progressive TM}
\label{sec:mutex-dsm}
We describe an implementation of a mutex object $L(M)$ from a strictly serializable, strongly
progressive TM implementation $M$ (Algorithm~\ref{alg:mutex-dsm}).
The algorithm is based on the mutex implementation in \cite{lee-thesis}.

Given a sequential implementation, we use a TM to execute the sequential code in a concurrent environment
by encapsulating each sequential operation within an \emph{atomic}
transaction that replaces each read and write of a t-object with the
transactional read and write implementations, respectively. 
If the transaction commits, then the result of the operation is
returned; otherwise if one of the transactional operations aborts.
For instance, in Algorithm~\ref{alg:mutex-dsm}, we wish to atomically read a t-object $X$, write a new value to it
and return the old value of $X$ prior to this write. To achieve this, we
employ a strictly serializable TM implementation $M$. 
Moreover, we assume that $M$ is strongly progressive, \emph{i.e.}, in every execution,
at least one transaction successfully commits and the value of $X$ is returned.

\vspace{1mm}\noindent\textbf{Shared objects.}
We associate each process $p_i$ with two alternating identities $[p_i,\ms{face}_i]$; $\ms{face}_i \in \{0,1\}$.
The strongly progressive TM implementation $M$ is used to enqueue processes that attempt to enter the critical section within
a single t-object $X$ (initially $\bot$).
For each $[p_i,\ms{face}_i]$, $L(M)$ uses a register bit $\ms{Done}[p_i,\ms{face}_i]$ that indicates if this face of the process
has left the critical section or is executing the $\lit{Entry}$ operation. Additionally, we use
a register $\ms{Succ}[p_i,\ms{face}_i]$ that stores the process expected to succeed $p_i$ in the critical section.
If $\ms{Succ}[p_i,\ms{face}_i]=p_j$, we say that $p_j$ is the \emph{successor of $p_i$} (and $p_i$ is the \emph{predecessor} of $p_j$). 
Intuitively, this means
that $p_j$ is expected to enter the critical section immediately after $p_i$.
Finally, $L(M)$ uses a $2$-dimensional bit array $\ms{Lock}$: for each process $p_i$, there are $n-1$ registers associated with
the other processes. For all $j\in \{0,\ldots , n-1\}\setminus \{i\}$, the registers $\ms{Lock}[p_i][p_j]$
are local to $p_i$ and registers $\ms{Lock}[p_j][p_i]$ are remote to $p_i$. Process $p_i$ can only access registers
in the $\ms{Lock}$ array that are local or remote to it.

\vspace{1mm}\noindent\textbf{Entry operation.}
A process $p_i$ adopts a new identity $\ms{face}_i$ and writes $\false$ to $\ms{Done}(p_i,\ms{face}_i)$ to indicate
that $p_i$ has started the $\lit{Entry}$ operation. Process $p_i$ now initializes the successor of $[p_i,\ms{face}_i]$
by writing $\bot$ to $\ms{Succ}[p_i,\ms{face}_i]$. Now, $p_i$ uses a strongly progressive TM implementation $M$
to atomically store its \emph{pid} and identity i.e., $\ms{face}_i$ to t-object $X$ and returns the \emph{pid}
and identity of its \emph{predecessor}, say $[p_j,\ms{face}_j]$. Intuitively, this suggests that
$[p_i,\ms{face}_i]$ is scheduled to enter the critical section immediately after $]p_j,\ms{face}_j]$ exits the critical
section.
Note that if $p_i$ reads the initial value of t-object $X$, then it immediately enters the critical section.
Otherwise it writes \emph{locked} to the register $\ms{Lock}[p_i,p_j]$ and sets itself to be the successor of $[p_j,\ms{face}_j]$
by writing $p_i$ to $\ms{Succ}[p_j,\ms{face}_j]$.
Process $p_i$ now checks if $p_j$ has started the $\lit{Exit}$ operation by checking if $\ms{Done}[p_j,\ms{face}_j]$
is set. If it is, $p_i$ enters the critical section; otherwise $p_i$ spins on the register $\ms{Lock}[p_i][p_j]$
until it is \emph{unlocked}.

\vspace{1mm}\noindent\textbf{Exit operation.}
Process $p_i$ first indicates that it has exited the critical section by setting $\ms{Done}[p_i,\ms{face}_i]$, following
which it \emph{unlocks} the register $\ms{Lock}[\ms{Succ}[p_i,\ms{face}_i]][p_i]$ to allow $p_i$'s successor to
enter the critical section.
\subsection{Proof of correctness}
\begin{lemma}
\label{lm:mutex}
The implementation $L(M)$ (Algorithm~\ref{alg:mutex-dsm}) satisfies mutual exclusion.
\end{lemma}
\begin{proof}
Let $E$ be any execution of $L(M)$.
We say that $[p_i,\ms{face}_i]$ is the \emph{successor} of $[p_j,\ms{face}_j]$ if $p_i$ reads the value of $\ms{prev}$
in Line~\ref{line:while} to be 
$[p_j,\ms{face}_j]$ (and $[p_j,\ms{face}_j]$ is the \emph{predecessor} of $[p_i,\ms{face}_i]$); 
otherwise if $p_i$ reads the value to be $\bot$, we say that $p_i$ has no predecessor.

Suppose by contradiction that there exist processes $p_i$ and $p_j$ that are both inside the critical section after $E$.
Since $p_i$ is inside the critical section, either (1) 
$p_i$ read $\ms{prev}=\bot$ in Line~\ref{line:tm}, or
(2) $p_i$ read that $\ms{Done}[\ms{prev}]$ is $\true$ (Line~\ref{line:entrydone2}) or $p_i$ reads that
$\ms{Done}[\ms{prev}]$ is $\false$ and $\ms{Lock}[p_i][\ms{prev.pid}]$
is \emph{unlocked} (Line~\ref{line:entryunlock}).

(Case $1$) Suppose that $p_i$ read $\ms{prev}=\bot$ and entered the critical section. Since in this case, $p_i$ does
not have any predecessor, some other process that returns successfully from the \emph{while} loop in Line~\ref{line:while}
must be successor of $p_i$ in $E$. 
Since there exists $[p_j,\ms{face}_j]$ also inside the critical section after $E$, $p_j$ reads that either
$[p_i,\ms{face}_i]$ or some other process to be its predecessor. Observe that there must exist some such 
process $[p_k,\ms{face}_k]$ whose predecessor is $[p_i,\ms{face}_i]$. Hence, without loss of generality, we can assume that
$[p_j,\ms{face}_j]$ is the successor of $[p_i,\ms{face}_i]$.
By our assumption, $[p_j,\ms{face}_j]$ is also inside the critical section. Thus, $p_j$
\emph{locked} the register $\ms{Lock}[p_j,p_i]$ in Line~\ref{line:entrylock} and set itself to be $p_i$'s successor
in Line~\ref{line:entrysucc2}. Then, $p_j$ read that $\ms{Done}[p_i,\ms{face}_i]$ is $\true$
or read that $\ms{Done}[p_i,\ms{face}_i]$ is $\false$ and waited until $\ms{Lock}[p_j,p_i]$ is \emph{unlocked}
and then entered the critical section. But this is possible only if $p_i$ has left the critical section and 
updated the registers $\ms{Done}[p_i,\ms{face}_i]$ and $\ms{Lock}[p_j,p_i]$ in Lines~\ref{line:exitdone}
and \ref{line:exitunlock} respectively---contradiction to the assumption that $[p_i,\ms{face}_i]$ is also
inside the critical section after $E$.

(Case $2$) Suppose that $p_i$ did not read $\ms{prev}=\bot$ and entered the critical section.
Thus, $p_i$ read that $\ms{Done}[\ms{prev}]$ is $\false$ in Line~\ref{line:entrydone2} and $\ms{Lock}[p_i][\ms{prev.pid}]$
is \emph{unlocked} in Line~\ref{line:entryunlock}, where $\ms{prev}$ is the predecessor of $[p_i,\ms{face}_i]$.
As with case $1$, without loss of generality, we can assume that
$[p_j,\ms{face}_j]$ is the successor of $[p_i,\ms{face}_i]$ or 
$[p_j,\ms{face}_j]$ is the predecessor of $[p_i,\ms{face}_i]$.

Suppose that $[p_j,\ms{face}_j]$ is the predecessor of $[p_i,\ms{face}_i]$, \emph{i.e.}, $p_i$ writes the value
$[p_i,\ms{face}_i]$ to the register $\ms{Succ}[p_j,\ms{face}_j]$ in Line~\ref{line:entrysucc2}.
Since $[p_j,\ms{face}_j]$ is also inside the critical section after $E$,
process $p_i$ must read that $\ms{Done}[p_j,\ms{face}_j]$ is $\true$ in Line~\ref{line:entrydone2}
and $\ms{Lock}[p_i,p_j]$ is \emph{locked} in Line~\ref{line:entryunlock}.
But then $p_i$ could not have entered the critical section after $E$---contradiction.

Suppose that $[p_j,\ms{face}_j]$ is the successor of $[p_i,\ms{face}_i]$, \emph{i.e.}, $p_j$ writes the value
$[p_j,\ms{face}_j]$ to the register $\ms{Succ}[p_i,\ms{face}_i]$. Since both $p_i$ and $p_j$ are inside the critical section
after $E$, process $p_j$ must read that $\ms{Done}[p_i,\ms{face}_i]$ is $\true$ in Line~\ref{line:entrydone2}
and $\ms{Lock}[p_j,p_i]$ is \emph{locked} in Line~\ref{line:entryunlock}.
Thus, $p_j$ must spin on the register $\ms{Lock}[p_j,p_i]$, waiting for it to be \emph{unlocked} by $p_i$
before entering the critical section---contradiction to the assumption that both $p_i$ and $p_j$ are inside the critical section.

Thus, $L(M)$ satisfies mutual-exclusion.
\end{proof}
\begin{lemma}
\label{lm:dead}
The implementation $L(M)$ (Algorithm~\ref{alg:mutex-dsm}) provides deadlock-freedom.
\end{lemma}
\begin{proof}
Let $E$ be any execution of $L(M)$.
Observe that a process may be stuck indefinitely only in Lines~\ref{line:tm} and \ref{line:entryunlock} as it performs
the \emph{while} loop.

Since $M$ is strongly progressive, in every execution $E$ that contains an invocation of $\lit{Enter}$
by process $p_i$, some process returns $\true$ from the invocation of $\ms{func}()$ in Line~\ref{line:tm}.

Now consider a process $p_i$ that returns successfuly from the \emph{while} loop in Line~\ref{line:tm}.
Suppose that $p_i$ is stuck indefinitely as it performs the \emph{while} loop in Line~\ref{line:entryunlock}.
Thus, no process has \emph{unlocked} the register $\ms{Lock}[p_i][\ms{prev.pid}]$ by writing to it in the $\lit{Exit}$ section.
Recall that since $[p_i,\ms{face}_i]$ has reached the \emph{while} loop in Line~\ref{line:entryunlock}, $[p_i,\ms{face}_i]$
necessarily has a predecessor, say $[p_j,\ms{face}_j]$, and has set itself to be $p_j$'s successor by writing
$p_i$ to register $\ms{Succ}[p_j,\ms{face}_j]$ in Line~\ref{line:entrysucc2}.
Consider the possible two cases: the predecessor of $[p_j,\ms{face}_j$ is some process $p_k$;$k \neq i$ or
the predecessor of $[p_j,\ms{face}_j$ is the process $p_i$ itself.

(Case $1$) Since by assumption, process $p_j$ takes infinitely many steps in $E$, the only reason that
$p_j$ is stuck without entering the critical section is that $[p_k,\ms{face}_k]$ is also stuck in
the \emph{while} loop in Line~\ref{line:entryunlock}. Note that it is possible for us to iteratively extend this execution
in which $p_k$'s predecessor is a process that is not $p_i$ or $p_j$ that is also stuck in the
\emph{while} loop in Line~\ref{line:entryunlock}. But then the last such process must eventually
read the corresponding $\ms{Lock}$ to be \emph{unlocked} and enter the critical section.
Thus, in every extension of $E$ in which every process takes infinitely many steps, some process will enter the critical section.

(Case $2$) Suppose that the predecessor of $[p_j,\ms{face}_j$ is the process $p_i$ itself.
Thus, as $[p_i,\ms{face}_]$ is stuck in the \emph{while} loop waiting for $\ms{Lock}[p_i,p_j]$ to be \emph{unlocked}
by process $p_j$, $p_j$ leaves the critical section, \emph{unlocks} $\ms{Lock}[p_i,p_j]$ in Line~\ref{line:exitunlock}
and prior to the read of $\ms{Lock}[p_i,p_j]$, $p_j$ re-starts the $\lit{Entry}$ operation,
writes $\false$ to $\ms{Done}[p_j,1-\ms{face}_j]$ and sets itself to be the successor of $[p_i,\ms{face}_i]$
and spins on the register $\ms{Lock}[p_j,p_i]$. However, observe that process $p_i$, which takes infinitely many steps by our assumption
must eventually read that $\ms{Lock}[p_i,p_j]$ is \emph{unlocked} and enter the critical section, thus establishing
deadlock-freedom.
\end{proof}
We say that a TM implementation $M$ \emph{accesses a single t-object} if in every execution $E$ of $M$ and
every transaction $T\in \ms{txns}(E)$, $|\Dset(T)| \leq 1$. We can now prove the following theorem:
\begin{theorem}
\label{th:mutex-tm}
Any strictly serializable, strongly progressive TM implementation $M$ that accesses a single t-object 
implies a \emph{deadlock-free}, 
\emph{finite exit} mutual exclusion
implementation $L(M)$ such that the RMR complexity of $M$ is within a constant factor of the RMR complexity of $L(M)$.
\end{theorem}
\begin{proof}
(Mutual-exclusion)
Follows from Lemma~\ref{lm:mutex}.

(Finite-exit)
The proof is immediate since the $\lit{Exit}$ operation contains no unbounded
loops or waiting statements.

(Deadlock-freedom)
Follows from Lemma~\ref{lm:dead}.

(RMR complexity)
First, let us consider the CC model. Observe that every event not on $M$ performed by a process $p_i$ 
as it performs the $\lit{Entry}$
or $\lit{Exit}$ operations incurs $O(1)$ RMR cost clearly, possibly barring the \emph{while} loop executed in
Line~\ref{line:entryunlock}. During the execution of this \emph{while} loop, process $p_i$ spins on the register
$\ms{Lock}[p_i][p_j]$, where $p_j$ is the predecessor of $p_i$.  Observe that $p_i$'s cached copy of $\ms{Lock}[p_i][p_j]$
may be invalidated only by process $p_j$ as it \emph{unlocks} the register in Line~\ref{line:exitunlock}.
Since no other process may write to this register and $p_i$ terminates the \emph{while} loop
immediately after the write to $\ms{Lock}[p_i][p_j]$ by $p_j$, $p_i$ incurs $O(1)$ RMR's.
Thus, the overall RMR cost incurred by $M$ is within a constant factor of the RMR cost of $L(M)$.

Now we consider the DSM model. As with the reasoning for the CC model, every event not on $M$ performed by a process $p_i$
as it performs the $\lit{Entry}$
or $\lit{Exit}$ operations incurs $O(1)$ RMR cost clearly, possibly barring the \emph{while} loop executed in
Line~\ref{line:entryunlock}. During the execution of this \emph{while} loop, process $p_i$ spins on the register
$\ms{Lock}[p_i][p_j]$, where $p_j$ is the predecessor of $p_i$. Recall that $\ms{Lock}[p_i][p_j]$
is a register that is local to $p_i$ and thus, $p_i$ does not incur any RMR cost on account of executing this loop.
It follows that $p_i$ incurs $O(1)$ RMR cost in the DSM model. Thus, the overall RMR cost of $M$ is within
a constant factor of the RMR cost of $L(M)$ in the DSM model.
\end{proof}
\begin{theorem}
\label{th:mutex-rmr}
(\cite{rmr-mutex})
Any deadlock-free, finite-exit mutual exclusion implementation from read, write and 
conditional primitives has an execution whose RMR complexity is $\Omega(n \log{n})$.
\end{theorem}
Theorems~\ref{th:mutex-rmr} and \ref{th:mutex-tm} imply:
\begin{theorem}
\label{th:mutex-tm-rmr}
Any strictly serializable, strongly progressive TM implementation from read, write and 
conditional primitives that accesses a single t-object has an execution whose RMR complexity is $\Omega(n \log{n})$.
\end{theorem}
%
%
\section{Related work and concluding remarks}
\label{sec:related}
Theorem~\ref{th:iclb} improves the read-validation step-complexity lower bound~\cite{GK09-progressiveness,tm-book} derived for \emph{strict-data
partitioning} (a very strong version of DAP) and (strong) invisible reads.
In a \emph{strict data partitioned} TM, the set of base objects used by the TM is split into 
disjoint sets, each storing information only about a single data item.
Indeed, every TM implementation that is strict data-partitioned satisfies weak DAP, but not vice-versa.
The definition of invisible reads assumed in \cite{GK09-progressiveness,tm-book}
requires that a t-read operation does not apply nontrivial events in any execution.
Theorem~\ref{th:iclb} however, assumes \emph{weak}
invisible reads, stipulating that t-read operations of a transaction $T$ do not apply nontrivial events only when $T$
is not concurrent with any other transaction.
%

The notion of weak DAP used in this paper was introduced by Attiya
\emph{et al.}~\cite{AHM09}.

Proving a lower bound for a concurrent object by reduction to a form of mutual exclusion has previously been used in~\cite{AlistarhAGG11,tm-book}.
Guerraoui and Kapalka~\cite{tm-book} proved that it is impossible to implement strictly serializable strongly progressive TMs
that provide \emph{wait-free} TM-liveness (every t-operation returns a matching response within a finite number of steps)
using only read and write primitives. 
Alistarh \emph{et al.} proved a lower bound
on RMR complexity of \emph{renaming} problem~\cite{AlistarhAGG11}.
Our reduction algorithm (Section~\ref{sec:rmr}) is inspired by the
$O(1)$ RMR mutual exclusion algorithm by Lee~\cite{lee-thesis}.

To the best of our knowledge, the TM properties assumed for Theorem~\ref{th:iclb}
cover all of the TM implementations that are subject
to the validation step-complexity~\cite{HLM+03,norec,DSS06}.

It is easy to see that the lower bound of Theorem~\ref{th:iclb} is tight for both strict serializability and opacity.
We refer to the TM implementation in~\cite{KR11-TR} or  \emph{DSTM}~\cite{HLM+03} for
the matching upper bound. 

Finally, we conjecture that the lower bound of Theorem~\ref{th:mutex-tm-rmr} is tight. Proving this remains an interesting open
question.
%

%
%
%
\bibliography{references}
\end{document}